\documentclass[letterpaper,twocolumn,11pt]{article}

\newcommand{\ignore}[1]{}

\usepackage{times}
\usepackage{amsfonts}
\usepackage{amsmath,amsthm}
\usepackage{latexsym}
\usepackage{graphicx}
\usepackage{subcaption}
\usepackage{xcolor}
\usepackage{xspace}
\usepackage{pgfplots}
\usepackage{enumitem}
\usepackage[colorlinks=true,citecolor=purple,urlcolor=blue]{hyperref}
\usepackage{framed}

\usepackage{etoolbox}
\newtoggle{fullversion}
\toggletrue{fullversion}

\renewcommand{\paragraph}[1]{\medskip\noindent\textbf{#1}}
\newtheorem{theorem}{Theorem}

\newcommand{\Rbalance}{{\em Rbalance}\xspace}
\newcommand{\NRbalance}{{\em NRbalance}\xspace}

\newcommand{\deq}{\mathrel{\mathop:}=}
\newlist{pseudocode}{enumerate}{2}
\setlist[pseudocode]{label=\raisebox{0.4px}{\scriptsize{\sf \arabic*:}}, ref=\arabic*, labelwidth=1em, labelsep=1em, align=right, itemsep=.2ex,topsep=0.5ex, leftmargin=1.3em}

\usepackage{microtype}

\title{ERC-20R and ERC-721R: Reversible Transactions on Ethereum}

\iftoggle{fullversion}
{
  \author{
    Kaili Wang \\ \textsf{\small kkwang@cs.stanford.edu}
  \and 
    Qinchen Wang \\ \textsf{\small qinchenw@cs.stanford.edu}
  \and 
    Dan Boneh \\ \textsf{\small dabo@cs.stanford.edu}
}}
{
  \date{}
  \author{}
}

\begin{document}

\maketitle

\begin{abstract} 
Blockchains are meant to be persistent: posted transactions are immutable and cannot be changed.
As a result, when a theft takes place, 
there are limited options for reversing the disputed transaction, and this has led to significant losses in the blockchain ecosystem. 

In this paper we propose reversible versions of ERC-20 and ERC-721, the most widely used token standards. With these new standards, a transaction is eligible for reversal for a short period of time after it has been posted on chain. 
After the dispute period has elapsed, the transaction can no longer be reversed. 
Within the short dispute period, a sender can request to reverse a transaction by convincing
a decentralized set of judges to first freeze the disputed assets, 
and then later convincing them to reverse the transaction. 

Supporting reversibility in the context of ERC-20 and ERC-721 raises many interesting technical challenges. 
This paper explores these challenges and proposes a design
for our ERC-20R and ERC-721R standards, 
the reversible versions of ERC-20 and ERC-721.
We also provide a
\iftoggle{fullversion}
  {\href{https://github.com/kkailiwang/erc20r}{prototype implementation}.} 
  {\href{https://github.com}{prototype implementation}.} 
Our goal is to initiate a deeper conversation about reversibility in the hope of reducing some of the
losses in the blockchain ecosystem.
\end{abstract}

\section{Introduction} \label{sec:introduction}
\vspace{-2ex}

Since their inception, cryptocurrencies have been plagued by thefts and accidental losses \cite{zandt_2022}.
Victims include 
end users~\cite{uniswapLP}, 
DAOs~\cite{daoHack}, 
bridges~\cite{RoninHack,harmony,polyHack,DeFiHacks}, 
and exchanges~\cite{ExHacks,bitgrailHack,coincheckHack,kucoinHack,goxHack}. 
Usually, the stolen assets are first transferred from the victim's address
to an address controlled by the attacker.  
From there the assets are laundered by transferring them to other addresses and eventually to an offramp.
In a few cases, the assets are seized at the offramp~\cite{Bitfinex}.

The annual losses can be quite high.  
In 2020, \$7.8 billion was stolen, and in 2021 that amount doubled to \$14 billion~\cite{chainalysis_2022,stasha_2022}. 
A recent well publicized attack on the Ronin bridge resulted in a theft of over \$600 million~\cite{RoninHack}. 
The attacker transferred ETH and USDC from the Ronin contract on Ethereum to an address controlled by the attacker.
From there, some of the funds were moved to the Tornado mixer~\cite{RoninTornado}.
Other bridges have experienced similar thefts~\cite{harmony,polyHack,DeFiHacks}.
In many of these attacks, the assets stolen were held in ERC-20 contracts.

Similarly, NFTs held in ERC-721 and ERC-1155 contracts have seen an uptick in thefts.
In a twelve months period following July 2021, an estimated \$100 million were stolen in NFTs
using phishing and other attacks~\cite{elliptic_2022,chainalysis_2022}. 

These attacks were often discovered soon after the theft took place.
Had there been a way to reverse the offending transaction(s) --
as in traditional finance --
the damage could have been greatly reduced. 

Beyond theft, transaction finality has also worked against us
when funds are accidentally sent to a wrong address. 
In May 2022, a Cosmos-based blockchain called JUNO passed a proposal to move \$36 million USD to a specific address.
The address contained a typo~\cite{typo} and consequently the funds were lost. 
The funds could have been recovered had there been a way to reverse that transaction. 

\paragraph{Reversible transactions.}
In a \href{https://twitter.com/vitalikbuterin/status/987262267036184577}{2018 tweet}, Vitalik Buterin wrote that 
\begin{quote}
Someone should come along and issue an ERC20 called ``Reversible Ether'' that is 1:1 backed by ether but has a DAO that can revert transfers within $N$ days.
\end{quote}
Nowadays, this can be applied to any ERC-20 token (not just wrapped ETH), as well as to NFTs.

Enabling reversible transactions is not easy and introduces many fascinating technical challenges.
The main contribution of this paper is to explore how to support reversibility within
the ERC-20 and ERC-721 framework.
We propose two new standards that allow transaction reversal within a limited time window,
say four days.
We call these standards ERC-20R and ERC-721R, the reversible versions of ERC-20 and ERC-721, respectively.

\medskip
We envision the following high-level workflow for reversing a posted transaction (see Figure~\ref{fig:api1}):
\begin{itemize}[itemsep=.5ex,topsep=.5ex,labelwidth=1em,leftmargin=1em]
\item {\em Request freeze.} The victim posts a freeze request to a governance contract, 
along with the relevant evidence, and some stake. 
A request to freeze a transaction can only be initiated by an address
that is directly affected by the transaction. 

\item {\em Freeze assets.} A decentralized set of judges decides to accept or reject the request.
If accepted, the judges instruct an on-chain governance contract to call the {\em freeze} function on the impacted ERC-20R or ERC-721R contract.
Subsequently, the assets in question are frozen and can no longer be transferred.
For NFTs, this is a simple matter of freezing the disputed NFT.
For ERC-20 tokens this is more complicated, as we explain below.
We discuss the operation of the governance contract,
and the selection of judges, in Section~\ref{sec:governance}. 
We envision the freeze process being relatively quick, 
taking the judges at most one or two days to make a decision.

\item {\em Trial.}  Both sides can then present evidence
to the decentralized set of judges.
Eventually the judges reach a decision, at which point
they instruct the governance contract to call 
either the {\em reverse} or {\em rejectReverse} functions 
on the impacted ERC-20R or ERC-721R contract. 
The {\em reverse} function transfers the disputed (frozen) assets to their
original owner.
The {\em rejectReverse} releases the freeze on the disputed assets
and leaves them where they are.
The trial may be lengthy, possibly taking several weeks or months. 
\end{itemize}

\begin{figure}
\begin{center}
  \includegraphics[width=0.9\linewidth]{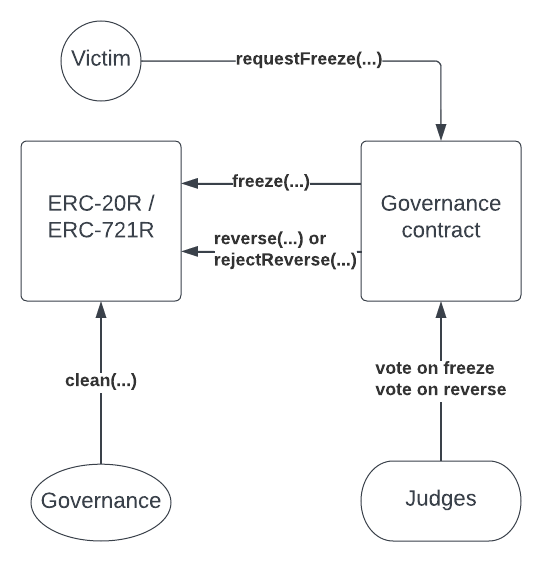}
  \caption{The process for reversing a transaction.}
  \label{fig:api1}
\end{center}
\end{figure}

\noindent
A more complete workflow is presented in Figure~\ref{fig:parties_involved}.

\begin{figure}[t!]
\begin{center}
  \includegraphics[width=0.9\linewidth]{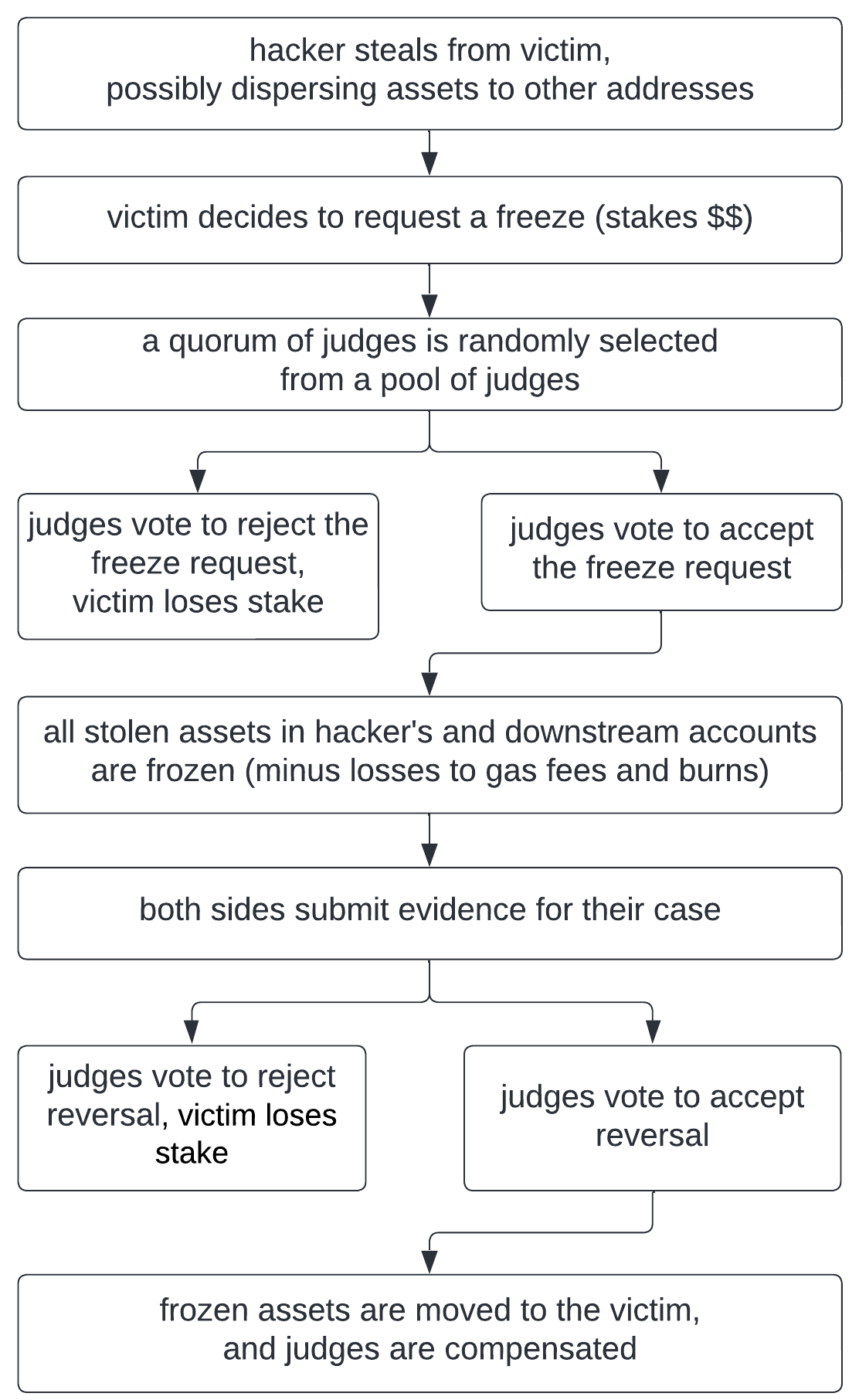}
  \caption{Overview of reversal process}
  \label{fig:parties_involved}
\end{center}
\end{figure}

\paragraph{Locating the stolen assets.}
By the time the victim submits a freeze request, 
the attacker may have already moved the stolen assets through multiple accounts.
In fact, the attacker can monitor the mempool, 
and move the assets as soon as it sees a request to freeze the stolen assets. 
In the case of an NFT, the attacker may have sold the stolen NFT to an unsuspecting honest user.
In the case of an ERC-20 token, 
the attacker may have divided the stolen tokens across multiple accounts;
it may have exchanged a portion of the tokens for another ERC-20 token using an honest on-chain exchange;
it may have burned a portion of the tokens;
or it may have sent the tokens to a mixer.
The new reversible standards must properly handle all these cases. 

In case of a dispute over an ERC-721 NFT, the freeze is applied to the current holder of the NFT:
either the original attacker, or an honest user who purchased the stolen NFT from the attacker.
If the judges decide that a theft took place, 
then the ERC-721R contract sends the NFT back to the pre-theft owner.
The current owner of the NFT loses the NFT.
This policy is consistent with tort law in many countries, but of course, 
other policies can be implemented.

In case of a dispute over stolen ERC-20 tokens things are more complicated.
By the time the freeze is executed, the funds may have been dispersed across many downstream accounts,
some honest and some dishonest.
In Section~\ref{sec:methodology} we present an example algorithm that assigns fractional responsibility to
each of the downstream accounts that received a portion of the stolen funds.
The partial freeze is then applied to these accounts.
Implementing this freeze strategy requires the ERC-20 contract to maintain a transaction log
during the dispute window so that the {\em freeze} function can trace the funds when
it is called by the governance contract.  
If the judges decide that a theft took place, the ERC-20R contract moves the frozen
tokens from the obligated accounts to the pre-theft account.
We discuss this in more detail in Section~\ref{sec:methodology}.

\paragraph{Will reversibility introduce delays?}
Suppose Alice holds tokens in an ERC-20R (an R-token)
and she wishes to exchange them for ETH or for some 
other ERC-20 (a non-R-token).
Bob is willing to do the ETH-for-token exchange with Alice
and they agree on the exchange rate.
Clearly Bob will not release his ETH to Alice
until he is assured that the R-tokens that Alice sent him 
cannot be taken back due an upstream reversal request.
This means that Bob will only accept R-tokens that were transferred
to Alice at least four days ago, if the dispute window is four days.
We refer to such R-tokens as ``old'' tokens,
since they are no longer subject a freeze and reversal. 
If Alice has sufficiently many old R-tokens then the transaction can settle
right away.  
Otherwise, Bob will delay final settlement for four days,
after which he is assured that Alice's R-tokens cannot be taken away
due to an upstream reversal. 
In particular, if Alice steals some R-tokens from Victor (the victim)
and immediately tries to exchange them for Bob's non-R-token, 
that exchange will only settle four days later,
by which time Victor may have asked for the theft to be reversed. 
To reiterate, an exchange of a recently transferred R-token for a
non-R-token will be subject to delay. 

Interestingly, if Alice wants to exchange one reversible token
for another reversible token, 
then Bob can immediately settle that exchange.
No delay is needed.
If the transaction that sends Alice's tokens to Bob is later reversed,
then Bob can request to reverse the matching transaction
that sends his tokens to Alice.
Hence an R-token to R-token exchange can settle instantly,
but an R-token to a non-R-token exchange may need to be delayed until
the reversible tokens are sufficiently old.  

The consequence of this is that once a few key tokens become reversible, 
there is a strong incentive for other tokens to become reversible
to avoid settlement delays. 
In other words, {\bf\em reversibility is viral.}

\subsection{Implementation}

We provide a reference implementation of these new standards. 
Our Solidity implementation is split into two parts: 
(i) the main ERC-20R and ERC-721R contracts that keep track of all the balances and transactions, and 
(ii) a governance contract that selects judges and gathers votes.
Our ERC-20R and ERC-721R implementations are extensions of the OpenZeppelin non-reversible contracts.

Recall that an ERC-20 contract manages the balances 
of many accounts. 
In our ERC-20R contract, an account balance is a pair of numbers
we call \Rbalance and \NRbalance.  
\begin{itemize}[topsep=1ex,itemsep=1ex,parsep=1ex]
\item 
The {\bf\NRbalance} is the current account balance due to incoming 
transactions whose dispute window has elapsed.
Funds in the \NRbalance are non-reversible: 
they are no longer subject to a potential freeze and reversal. 

\item 
The {\bf\Rbalance} is the account balance due to recent incoming transactions.
Funds in the \Rbalance are subject to reversal. 
\end{itemize}
As the dispute window elapses, funds move from the \Rbalance
to the \NRbalance.
This is done using the {\em clean} function discussed below. 

When an account owner sends funds from its own account to another account
(say, to fulfill an exchange of assets),
the account owner specifies how much to take
out of the \Rbalance and how much to take out of the \NRbalance.  
That is, the standard ERC-20 {\em transfer} function now comes in two flavors:
{\em transfer()} and {\em Rtransfer()}.
The former transfers from the non-reversible balance,
while the latter transfers from the reversible balance.
Either way, the transferred funds are added to the \Rbalance of the recipient. 
Note that our {\em transfer} function is backwards compatible with the ERC-20 specification. 
When the contract needs to burn tokens in an account 
(e.g., due to a fiat or base token withdrawal),
it will typically only burn tokens from the \NRbalance of the account.
However, there may be situations where the contract is willing 
to burn from the \Rbalance. 

\ignore{
An ERC-20R contract must enforce two rules:
\begin{itemize}[topsep=1ex,itemsep=1ex,parsep=1ex]
\item {\bf The transfer rule:}
transferred funds using either {\em transfer()} or {\em Rtransfer()}
are always added to the \Rbalance of the recipient. 

\item {\bf The burn rule:}
when burning a token (e.g., due to redemption for fiat),
the burn amount always comes out of the \NRbalance of the account.
\end{itemize}
The burn rule is due to the irreversibility of a burn.
}

\medskip\noindent
Beyond the new transfer interface,
an ERC-20R and ERC-721R contract exposes
the new interface functions 
{\em freeze}, {\em reverse}, {\em rejectReverse}, and {\em clean}.
Let us describe this API in more detail:
\begin{itemize}[itemsep=.5ex,topsep=0ex,parsep=.5ex,partopsep=0em,leftmargin=1em,labelwidth=1em]
    \item {\em freeze}(): calculates the amounts to freeze on the attacker's address as well as potential downstream addresses, and freezes those amounts. 
    For ERC-20R  
    it returns a $\textit{claimID}$
    that points to an on-chain list of (account,amount) pairs.
    The list identifies all the accounts that contain frozen assets associated with the complaint, and the amount frozen in each. 
    For ERC-721R it 
    returns a boolean success flag.
    The inner workings of the {\em freeze} function is explained in Section~\ref{sec:methodology}. 
    
    \item {\em reverse}(): sends all frozen assets associated with the claimed theft back to the original owner. 
    For ERC-20R, takes as argument a valid \textit{claimID}.
    For ERC-721R, takes in arguments $(\textit{tokenId}, \textit{index})$
    where $\textit{index}$ identifies the transaction being reversed.
    
    \item {\em rejectReverse}(): unfreezes all amounts associated with the claim. 
    For ERC-20R, takes as argument a \textit{claimID}; 
    for ERC-721R, takes as argument a \textit{tokenId}.
    
    \item {\em clean}(): Reversible contracts store some transaction data on chain.  
    The {\em clean} function removes on-chain information for transactions whose dispute window has elapsed. 
    The data structure that is being updated is explained in
    Section~\ref{sec:freeze}. 
    In addition, for ERC-20R this function moves the relevant balance 
    from the \Rbalance to the \NRbalance. 
\end{itemize}
The {\em freeze}, {\em reverse} and {\em rejectReverse} functions can only be called by the governance contract,
while anyone can call the {\em clean} function.

\subsection{Related Work} \label{sec:preliminaries}
 
In the fall of 2018, one of the co-creators of the ERC-20 standard proposed the concept of a {\em reversible ICO}, 
where investors would be able to get a refund amount inversely proportional to how recently they invested~\cite{rico}. 
Although this safeguards against a single token being a scam at launch, 
it does not protect against malicious transactions. 

Eigenmann~\cite{githubreverse} drafted a contract for a reversible token that extends the ERC-20 standard. 
It used an escrow method, where the escrow period was 30 days, 
during which the sender could recall the money at any time.
This is problematic because Bob could pay Alice for a service,
and then reverse the payment 28 days later, after Alice completed the service.
A similar approach is used in a proposal for refunds in ERC-721 mints~\cite{cryptofighters},
where Bob can get his money back 
within a certain time window after buying an NFT.
In our proposal, Bob would need to present sufficient evidence to a committee of judges for a transaction to be reversed. 
This protects counter-parties to the disputed transaction. 

The {\em Reversecoin} project from 2015 launched as a layer~1 blockchain~\cite{challa}. 
It introduced a timeout period between transaction initiation and confirmation. 
Each account has an offline key pair that enables the owner to either reverse a transaction or immediately confirm it. 
This may not prevent some modern hacks: the attacker would either steal
the confirmation key, or trick the user into using the confirmation key to confirm a malicious transaction. 
The elegant Bitcoin Convenants proposal~\cite{covenants} similarly uses 
two keys (or more) to enable a vault owner to finalize or revert transactions 
from the vault 24 hours after they were posted.

A closely related project is Lossless.io~\cite{lossless}.
The company provides a wrapped version of certain ERC-20 tokens.
Anyone with staked LSS tokens can monitor on-chain events for hacks, and can freeze an address if the address is involved in a hack. 
Only a single address can be frozen for a given hack, 
so one must act quickly after a hack to freeze the attacker's address
before the funds are further dispersed. 
The quickest spotter is rewarded in a winner-takes-all fashion. 
Transactions can also be frozen automatically by open-source mechanisms 
built by the Lossless team. 
The company then decides unilaterally whether the reversal is warranted,
and if so reverses the transfer. 
If a recipient burns the wrapped token it receives,
by converting them back to the base unwrapped token,
then reversal is no longer possible.
This cannot happen in our proposal since token burns come from an account's \NRbalance
(unless the contract willingly accepts the risk of burning \Rbalance tokens).

Finally, centralized exchanges maintain the ability to freeze and remove assets.
For example, {\em Binance USD} (BUSD), issued by Paxos, states that  (\href{https://www.binance.com/en/blog/futures/busd-all-you-need-to-know-about-the-stablecoin-421499824684903051}{link}):
\begin{quote}
Paxos also has the ability to create and burn BUSD tokens at will, as well as freeze and remove funds from people who exhibit nefarious or illicit activity.
\end{quote}
The same holds for other centralized tokens.
In these centralized systems, 
the operator acts as a centralized judge that can reverse transactions.

\section{The transaction reversal process} 
\label{sec:methodology}

In this section we describe the details of the ERC-20R and ERC-721R process of reverting a transaction.
We describe the data structures and algorithms needed
to track the funds that will be frozen and later sent back to the victim, if the victim prevails.

\subsection{Freezing assets}
\label{sec:freeze}

After a theft from an ERC-20R or an ERC-721R contract, 
the victim posts an on-chain freeze request to the governance contract.  
The request includes the offending transaction ID,
a link to evidence that an unauthorized transfer took place,
and some stake.
If the judges are convinced by the evidence (see Section~\ref{sec:governance})
then they instruct the governance contract to call
the {\em freeze} function on the relevant ERC-20R or ERC-721R contract.
Once the assets are frozen, any attempt to transfer them will fail.
If the judges are unconvinced, they instruct the governance contract to 
reject the request, and the victim loses the stake.

In this section we explain what happens when the
governance contract calls the {\em freeze} function on the relevant ERC-20R or ERC-721R contract.

\subsubsection{An ERC-721R freeze}

The {\em freeze} function on an ERC-721R contract is quite simple.
First, we add two structures to the contract:
{\small
\begin{verbatim}
 mapping(uint256 => bool )  _frozen;
 mapping(uint256 => Queue)  _owners;
\end{verbatim}
}
The {\tt \_frozen} structure indicates if a particular {\tt tokenID} is frozen.
If {\tt \_frozen[tokenID]} is true then the asset is frozen, and cannot be transferred.  

\newcommand{\TID}{\textit{tokenID}}
\newcommand{\OID}{\textit{owner}}
\newcommand{\BLK}{\textit{bn}}

The {\tt \_owners} structure keeps track of the recent list of owners for the asset identified by {\tt tokenID}.
On every transfer, a pair  {\it (newOwner,\ blockNumber)} 
is appended to the queue at {\tt \_owners[tokenID]} 
to record that at this block number, the new owner of {\tt tokenID} became {\it newOwner}.
This {\tt \_owners} structure looks as:
{\small
\begin{equation*}\label{eq:owners}
    \begin{cases}
      \TID_0 \to \bigl[(\OID_0,\BLK_0),(\OID_1,\BLK_1), \ldots \bigr] \\[1ex]
      \TID_1 \to \bigl[(\OID_0,\BLK_0),(\OID_1,\BLK_1), \ldots \bigr] \\[1ex]
      \TID_2 \to \bigl[(\OID_0,\BLK_0),(\OID_1,\BLK_1), \ldots \bigr] \\
      \quad \ldots
    \end{cases}
\end{equation*}
}%
where for each $\TID$ we have $\BLK_0 \leq \BLK_1 \leq \cdots$

\medskip\noindent
The governance contract freezes an asset by calling \\
\mbox{}\qquad  \texttt{freeze(tokenID, index)} \\
where {\tt index} points to an entry in the queue {\tt \_owners[tokenID]}.
The function first verifies that the disputed transfer took place within the dispute window
by using the block number at position {\tt index+1} in the queue {\tt \_owners[tokenID]}.
The function also ensures that the asset is not already frozen.
If the checks succeed, the function sets {\tt \_frozen[tokenID]} to true, 
indicating that the asset is frozen.
That's it.
Doing these checks on-chain ensures that even malicious judges cannot
reverse a transaction outside of the dispute window. 

\medskip
If at a later time the ERC-721R contract is asked to revert the disputed transfer
(that is, the governance contract calls the {\em reverse} function), 
the contract simply transfers {\tt tokenID}
to the address written at position {\tt index} in the queue {\tt \_owners[tokenID]}.
This was the owner before the disputed transaction. 

\medskip
Calling the {\em clean} function on this contract
with a list of token IDs
removes data from the {\tt \_owners} structure that is no longer needed.
This removes transfers that can no longer be reverted
because the dispute window has elapsed.
It is done to save on chain storage.
Note that a frozen tokenID cannot be cleaned until it becomes unfrozen;
this preserves the data needed to reverse a disputed transaction, if needed.

\subsubsection{An ERC-20R freeze}

The freeze function on an ERC-20 contract is much more complicated. 
The problem is that the tokens might have been transferred to multiple accounts between the time of the theft and the freeze request.

Figure~\ref{fig:theftGraph} shows a transaction from a victim address $v$ to an attacker address $a_0$.
Subsequent transfers from $a_0$ that took place after the disputed transaction, but prior to a freeze request, are indicated as directed edges in the graph.  All these downstream addresses may hold stolen funds that may need to be frozen.

\begin{figure}[t!]
  \includegraphics[width=\linewidth]{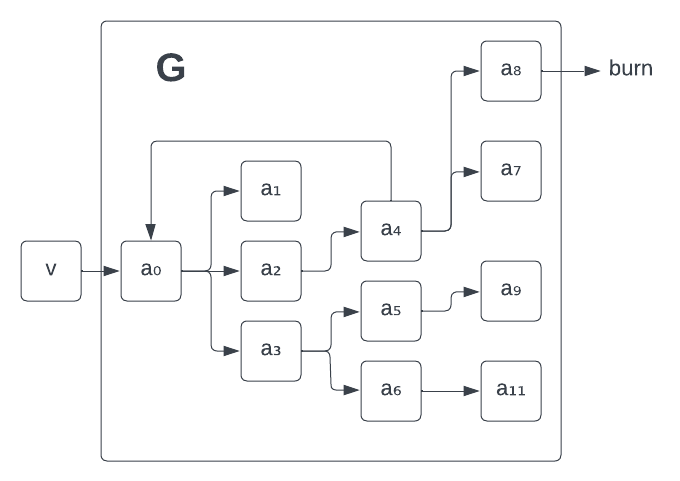}
  \caption{An example graph $G$ of transfers following the disputed transaction from $v$ to $a_0$.}
  \label{fig:theftGraph}
\end{figure}

Another complication is that tokens can be burned
and simply vanish from the system.
For example, in a stablecoin contract that is collateralized
by a fiat currency, the account $a_8$ may choose to redeem
its coins for fiat currency, at which point the coins at $a_8$
are burned.  
If the contract supports burning token from the \Rbalance of $a_8$
then there might not be sufficient assets to freeze
at address $a_8$, and the stolen funds are permanently lost. 
This is why burns generally apply to the \NRbalance of an account. 

Yet another complication is that a single address, 
say address $a_1$,
might own funds obtained from multiple thefts.  
When the funds are frozen, the system needs to remember
the amount frozen for each reversion request,
and if a reversion request is approved,
only the funds associated with that request should be
taken out of account $a_1$. 

If there are multiple freeze events on a single account,
the frozen amount is cumulative.
In particular, if account $a$ has a total of $s$ frozen coins,
then the ERC-20R contract will reject any transfer that
causes the \Rbalance of account $a$ to drop below $s$.

\paragraph{Calculating suspect addresses.}
Let us begin by describing a specific strategy for 
calculating the amount to freeze on each account
in the graph~$G$ from Figure~\ref{fig:theftGraph}.
Accounts that are subject to a freeze are called
{\em suspect} addresses.
Our freeze algorithm will chase the stolen funds across 
the transaction graph and freeze funds across suspected addresses.
The algorithm will freeze assets that are as close as 
possible to~$a_0$ in~$G$.

\medskip
Before we describe the detailed algorithm, 
let us first see a few examples.
Suppose that $s$ coins are stolen from the victim~$v$
and transferred to account $a_0$ at time $T_0$.
At time $T_f > T_0$ the ERC-20R contract receives a request
to freeze that transaction. 

\smallskip\noindent Example 1. 
If at time $T_f$ the \Rbalance at $a_0$ is $s$ or more,
then~$s$ coins will be frozen at $a_0$ and the process concludes.

\smallskip\noindent Example 2. 
Suppose that at some time between $T_0$ and $T_f$ a transaction transfers $s/4$ coins from $a_0$ to $a_1$
and an additional $s/4$ coins from $a_0$ to $a_2$. 
The remaining \Rbalance at $a_0$ is $s/2$ and no subsequent transactions apply to $a_0$.  
Then at the freeze time, the entire $s/2$ \Rbalance at $a_0$ will be frozen.
Moreover, if the \Rbalance at $a_1$ and $a_2$ is exactly $(s/4)$, 
then $(s/4)$ coins will be frozen at each of those accounts.
Now that a total of $s$ coins has been frozen, the process terminates.
Note that $a_1$ or $a_2$ might be the addresses of an honest exchange or a mixing service,
in which case a portion of the coin's liquidity pool at the exchange or mix 
will become frozen.

\smallskip\noindent Example 3. 
Suppose that following the disputed transaction from $v$ to $a_0$,
there is a second transaction $a_0 \to a_1$.
Clearly if there are insufficient funds at $a_0$ at the time of the freeze, 
then some of the freeze obligations should pass on to $a_1$, as explained in the previous paragraph.
However, suppose that at some time {\em before} the $a_0 \to a_1$ transaction, there is a transaction $a_1 \to a_2$
that transfers funds from~$a_1$ to~$a_2$. 
We propose that none of the freeze obligations will pass on to $a_2$, 
because the $a_1 \to a_2$ transaction was posted before the disputed funds arrived at $a_1$.
The funds sent to $a_2$ are not directly involved in the dispute.
One can give examples where this policy can lead us astray,
but by in large, we claim that $a_2$ should not be involved in the freeze.

\medskip
These examples suggest that the freeze process is an iterative procedure
that freezes the maximum amount possible at every step.
If the \Rbalance at the current node is insufficient,
then the obligation is passed to the descendants of that node.
The process terminates once $s$ coins are frozen,
or once there are no more descendants to process. 

\medskip
We stress that the freezing algorithm in its entirety runs in a {\em single} transaction.
This ensures that account balances cannot change while the freeze process is executing.

\newcommand{\frozen}{\textit{toFreeze}}
\newcommand{\src}{\textit{src}}
\newcommand{\dest}{\textit{dest}}
\newcommand{\val}{\textit{val}}
\newcommand{\oblig}{\textit{oblig}}
\newcommand{\tim}{T}
\newcommand{\edges}{E}
\newcommand{\FT}{t_{\scriptscriptstyle \text{f}}}

\paragraph{Terminology.}
Suppose that the freeze function is called to freeze a transaction $t_0$
that transferred $s$~coins from address~$v$ to address~$a_0$.
Let $\FT$ be the posted freeze transaction on chain.
To describe the freeze algorithm, we use the following notation:
\begin{itemize}[itemsep=1ex,topsep=1ex,parsep=1ex]
\item 
$\frozen(a)$ is the number of coins that the freeze transaction $\FT$
will freeze at address $a$.  
At the start of the algorithm $\frozen(a) = 0$ for all~$a$, with the exception of $a_0$, where
\begin{equation} \label{eq:freeze}
  \frozen(a_0) \deq \min\bigl(s, \textit{Bal}(a_0)\bigr).
\end{equation}
The quantity $\textit{Bal}(a)$, for an address $a$, 
is the available \Rbalance at $a$ at the time of $\FT$.
This $\textit{Bal}(a)$ is calculated as the \Rbalance at $a$ at the beginning of the freeze transaction
minus the amount of coins already frozen at $a$
due to a prior dispute.
Thus~\eqref{eq:freeze} will freeze the maximum amount possible at $a_0$. 

\item For a transaction $t = (a \to b)$, from $a$ to $b$, 
let $\val(t)$ be the value transferred from $a$ to $b$.

\item $\textit{burnedAt}(a)$ 
is set to the number of coins burned at address $a$ from its \Rbalance
between the first transaction that sent a portion of the disputed funds 
to $a$ and the freeze transaction~$\FT$.
\end{itemize}

Now, consider the graph~$G$ that is defined by the set of transactions that took place after the disputed transaction
and before the freeze transaction.
The nodes in the graph are addresses, and every directed edge represents a transfer from one address to another.
The graph only includes an edge $b \to c$ if there is directed path of transactions from $a_0$ to address $b$ that all took place after $t_0$. 

\paragraph{The algorithm.}
We first describe a freeze algorithm that applies when the
graph~$G$ rooted at $a_0$ is a directed {\em acyclic} graph (DAG).   
In Appendix~\ref{sec:analysis} we extend the algorithm 
to handle cycles by introducing a pre-processing phase 
that eliminates cycles. 

The freeze algorithm is implemented 
in the function $\textit{CalcFreeze}$ shown in Figure~\ref{fig:code}.
This function is called as 
\[   \textit{CalcFreeze}(t_0)   \]
where $t_0$ is the disputed transaction. 

The algorithm begins by constructing a \href{https://en.wikipedia.org/wiki/Topological_sorting}{topological sort}
of the vertices of~$G$ staring at $a_0$.
A topological sort is a list $L$ of the vertices in $G$,
where every vertex in $G$ appears exactly once in $L$,
such that for every edge $(a \to b)$, the vertex $a$ appears in $L$ before $b$. 
Every DAG has a topological sort $L$,
and $L$ can be constructed in linear time in the number of edges.

For each address $a$ in $G$ the algorithm builds a value
$\oblig(a)$ that indicates the obligation amount that 
is passed to address $a$ and its descendants 
as a result of the disputed transaction.
The value of $\oblig(a)$ can increase whenever the algorithm
process an address that sent funds to $a$. 
For all $a$, the array $\oblig(a)$ is initially set to zero.

\begin{figure}[h!]
\begin{framed}
\noindent  $\textit{CalcFreeze}(t_0)$: 
\quad\quad /\!\,/ freeze trans. $t_0 {\scriptstyle = (v \to a_0)}$ \\[1ex]
\ \ $L \deq \text{(topological sort of the graph rooted at $a_0$)}$ \\
\mbox{} \qquad\qquad /\!\,/ assuming the graph $G$ is a DAG \\[1ex]
\ \ $\oblig(a_0) \deq \val(t_0)$ \\
\mbox{} \qquad\qquad /\!\,/ the obligation of $a_0$ due to $t_0$ \\[1ex]
for each $a$ in $L$ in order do:  \qquad /\!\,/ start at $a_0$

\begin{pseudocode}
\item  \label{code:total}
   $\tau \deq \oblig(a)$ \\
\mbox{}       \qquad /\!\,/ total obligation at $a$ from parents

\item  $\frozen(a) \deq \min\bigl( \tau,\ \textit{Bal}(a) \bigr)$  \\
\mbox{}       \qquad /\!\,/ amount to freeze at $a$

\item   \label{code:amount}
   $\tau' \deq \tau - \frozen(a) - \textit{burnedAt}(a)$   \\
\mbox{}    \qquad   /\!\,/ the amount left to freeze, but do not \\
\mbox{}    \qquad   /\!\,/ pass burned amount downstream

\item if $\tau' \le 0$: continue
\quad /\!\,/ all done with $a$

\item \label{code:edges}
   $\edges(a) \deq \text{sort}(\ \{ t {\scriptstyle = (a \to b)} \}\ )$ \\[0.5ex]
    \mbox{}\qquad   /\!\,/  the set of trans. from $a$ sorted \\
    \mbox{}\qquad   /\!\,/  in reverse chronological order
 
\item  \label{code:loop}
for each $t {\scriptstyle = (a \to b)}$ in $\edges(a)$ in order do: \\[0.5ex]
    \mbox{}\qquad\quad   /\!\,/  for each outgoing edge from $a$, \\
    \mbox{}\qquad\quad   /\!\,/  starting with the most recent one
    
\item  \label{code:min}
  \quad\  $\textit{ob} \deq \min\bigl( \tau', \val(t) \bigr)$

\item  \label{code:sub}
  \quad\ $\oblig(b)  \  \text{+=}\  \ \textit{ob} \quad;\quad  \tau' \ \text{{\bf --}\,=}\  \textit{ob}$ \\[.5ex]
      \mbox{}\qquad\quad   /\!\,/  obligate recipient $b$ to $\textit{ob}$ tokens
      
\item \label{code:out}
  \quad\ if $\tau' \leq 0$: break   \\
      \mbox{}\qquad\quad /\!\,/ all done with $a$,  \\
      \mbox{}\qquad\quad /\!\,/ terminate the inner loop on line~\eqref{code:loop}
\end{pseudocode}
\vspace{-1ex}
\end{framed}
\vspace{-3ex}
\caption{The freeze algorithm for a DAG}
\label{fig:code}
\end{figure}

\medskip
The algorithm attempts to freeze the maximal amount 
possible at every node $a$, starting with the root~$a_0$.
Line~\eqref{code:amount} calculates $\tau'$, 
the amount left to freeze after the available \Rbalance at $a$ is frozen.
Moreover, the algorithm subtracts the \Rbalance amount burned at~$a$ 
because that obligation should not transfer to the descendants of~$a$.
The remaining amount, $\tau'$, should pass as an obligation to the descendants of~$a$. 
In Line~\eqref{code:loop} the algorithm loops over all the transactions that take funds out of~$a$.
The loop proceeds in reverse chronological order, 
namely from the most recent transaction to the oldest transaction.
Then Line~\eqref{code:sub} obligates the recipient of the most recent transaction from $a$ to the maximal possible amount.
This continues until all of $\tau'$ is passed as an obligation to the children of~$a$.
We analyze this algorithm and its properties in Appendix~\ref{sec:analysis}.

Note that every address $a$ is only visited once,
and only after all its parents have been processed. 
Consequently, the running time is linear in the number of edges in the graph.
More precisely, if there are $V$ nodes and $E$ edges in the graph then the running time is $O(V + E)$
thanks to the data structure we use that keeps edges in a chronologically sorted order.

\medskip
Once this algorithm completes, the contract will add the quantity $\frozen(a)$ to the number of coins frozen at address $a$
(recall that $a$ might already have frozen coins due to a prior dispute).
This means that a subsequent transfer out of address $a$ will
fail if it causes the \Rbalance of $a$ to drop below the commulative frozen amount.

\newcommand{\spenditures}{\textit{\_spenditures}}

\paragraph{Implementing the algorithm.} 
The ERC-20R contract needs to maintain enough state to support the freeze process.
We introduce a new Solidity data structure called $\spenditures$
that serves two functions: 
(i) when asked to freeze a transaction, the contract needs to verify that the transaction took place within the dispute window, and
(ii) for an address~$a$, the contract needs to identify all downstream addresses that received funds from $a$ after the disputed
transaction took place. 

The new Solidity data structure called $\spenditures$ is illustrated in Figure~\ref{fig:spenditures}.

\begin{figure}[!h]
  \includegraphics[width=\linewidth]{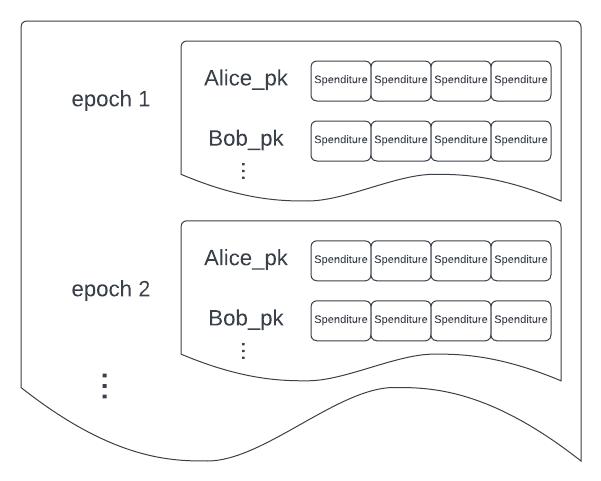}
  \caption{Spenditures nested map.}
  \label{fig:spenditures}
\end{figure}

We refer to each sequence of $\Delta$ blocks as an epoch 
(e.g., $\Delta = 1000$). 
For each epoch, the $\spenditures$ data structure lists all the transactions that took place during that epoch for each source address
(Alice and Bob in the figure).
Each Spenditure entry in Figure~\ref{fig:spenditures} contains the triple 
{\it (to, amount, blockNumber)}.
Every time the contract processes a transfer request,
it appends a Spenditure entry to the appropriate array
in the $\spenditures$ structure.
A burn transaction adds a Spenditure triple where the {\it to} field is null.

\medskip\noindent
The freeze function takes in arguments \\
\mbox{} \qquad {\it (epoch, from, index)}, \\
where {\it epoch} identifies the epoch in which the disputed transaction took place,
{\it from} identifies the payer, and 
{\it index} identifies the specific Spenditure in the data structure.
The {\it index} argument is provided by the governance contract by running an off-chain search over the $\spenditures$ structure to locate the
disputed transaction. 

The freeze function checks that the requested transaction is within the dispute window, and if so, runs the freeze algorithm using the $\spenditures$ data structure to determine the edges of the graph $G$ from Figure~\ref{fig:theftGraph}.
The final calculated freeze amounts per address are frozen.  
In addition, all these freeze amounts are recorded in a secondary array called {\tt \_claims} that is indexed by 
a newly generated 256-bit {\it ClaimID}.
If the victim prevails in trial, this array is used to transfer the correct amount from every suspect address to the victim.
The freeze function returns the {\it ClaimID} that indicates how to reverse the transaction.

\paragraph{The clean function.}
To minimize on-chain storage, our ERC-20R includes a {\tt clean} function that takes in an epoch and an array of addresses.
The function deletes all array entries in $\spenditures$ with the specified epoch and sender addresses 
for epochs for which the dispute window has elapsed. 
For every deleted entry, the specified $\spenditures$ amount
is transferred from the account's \Rbalance to its \NRbalance,
while making sure that the total \Rbalance does not drop below 
the frozen amount at the account. 
Anyone can call the clean function and free up on-chain storage that is no longer needed.

\subsection{Reversing a transaction}
\label{sec:revert}

Once the trial over a disputed transaction concludes, the governance contract will 
call the {\em rejectReverse} function or the {\em Reverse} function on the 
ERC-20R or ERC-721R contract.   

\begin{itemize}[topsep=1ex,itemsep=1ex,parsep=1ex]
\item An ERC-20R reversal takes as input a {\it ClaimID} which is an index into the {\tt \_claims} array.
The entry indicates the obligation of each suspect account in the disputed transfer.
The function transfers the specified amount from 
the \Rbalance of suspect accounts to the original owner, and clears the freezes.

\item An ERC-721R reversal takes as input {\it TokenID} and an {\it index} into {\tt \_owners} array.
This indicates the owner prior to the disputed transaction.  
The asset is then transferred to the original owner, and the asset is unfrozen.
\end{itemize}

\section{The Arbitration Process} \label{sec:governance}

We now turn to the governance process
where judges examine the evidence on a disputed transaction
and decide if the transaction should be revered. 
Some of the key issues in designing this process are:
(i) how judges are selected, 
(ii) how judges are compensated, and
(iii) how to discourage misbehaving judges, such as judges who take bribes or make bad decisions on disputed transactions.

This process is orchestrated by a governance contract.  
A single governance contract can govern many 20R and 721R contracts.
Judges vote on cases by either using an on chain or off chain voting system. 
Once enough votes on a case are collected, 
the governance contract calls the appropriate function on the 20R or 721R contract
to either execute the reversal or dismiss the case. 
To ensure that judges make independent decisions,
it is important that votes remain secret until sufficiently many votes are cast.
This could be done, for example, using some flavor of a commit-and-reveal voting scheme,
or any other semi-private voting scheme.

\subsection{Selecting judges}

We envision a large pool of available judges who will be compensated for their work. 
When a freeze request is submitted, the governance contract
selects a {\em random and unpredictable} quorum of $n$ judges from the pool.
The value of $n$ can be fixed, say $n=12$, or it can increase with the size of the 
disputed transaction, so that deciding on a large transaction requires more judges. 
This random selection is best implemented using a randomness beacon~\cite{beacon_chain}.
This quorum of judges decides on the initial freeze request, 
and later decides whether to approve or reject the reversal request.

Who can join the pool of judges and preside over cases?
One can envision a process that requires a real world identity, 
a requirement for qualifications, and a statement of conflicts,
much like real-world judge selection.
We will leave the details of how to admit applicants into the pool
of judges to future policy discussions of the exact mechanism.

\subsection{Compensating judges}

Every freeze request to the governance contract is accompanied by some stake from the party requesting the freeze.
This stake can be used to compensate judges for their work.
We state the following principle in deciding how to compensate judges:
When judges submit their vote on a freeze request or on a reversal decision, 
they are compensated for their work.   
However, {\em the compensation amount is independent of their voting decision.} 

If judges approve a freeze request, then the claimant's staked amount, minus fees, remains locked in the governance contract.
    If later the claimant loses the trial, the staked amount can compensate the defendant for their effort.
    If the claimant wins the trial, they could get back their excess stake. 
Importantly, if the judges reject the initial freeze request, 
then the claimant's staked amount minus the judges' fixed fees, should be burned.

\paragraph{Priority fees.}
While victims must provide a minimal stake along with a freeze request, they can optionally provide additional stake if they so choose.
This extra ``tip'' from the victim could potentially indicate the priority of the case. 
Judges could rule on cases by priority rather than by chronological order.
The tip could either be burned, or given to the prevailing party, or some combination of the two.

\subsection{Discouraging judicial misbehavior}

A judge might fail to vote in a timely manner, or they might frequently vote with the extreme minority,
potentially indicating an issue with their decision making.
In either case, the governance contract could remove such a judge from the pool.

A more interesting question is how to prevent bribery, 
where a party bribes judges to vote in its favor.
There are several technical solutions that could help partially mitigate this issue.
One approach is to use a secrecy-preserving process 
for selecting the quorum of judges to preside over a case:
a selected judge will learn that they were selected;
however, on their own, they will not know who the other judges are,
nor will they be able to prove to anyone that they were selected
(if they honestly follow the selection protocol).
Once all the selected judges vote, the set of judges is revealed 
along with a proof that the correct set voted. 
Importantly, the proof is revealed by the posted data from the selected judges,
not by a set of trustees. 
This way, a party who wishes to bribe a judge will not know who to bribe
because the set of judges is only revealed after they all voted. 
This makes it more difficult (but not impossible) to bribe judges.

\section{Discussion and Extensions} \label{sec:discussion}

\paragraph{Backwards compatibility.}
Our ERC-20R standard is backwards compatible with the existing DeFi infrastructure,
such as exchanges and lending protocols. 
First the {\em transfer()} function in ERC-20R has the same API
as the {\em transfer()} function in the ERC-20 standard. 
Hence, no software changes are needed to process a customer request.
Second, because the {\em transfer()} function transfers funds 
from Alice's \NRbalance 
(her balance of ``old'' tokens  that are no longer reversible),
the exchange is assured that the funds that Alice sends to it
will not be reversed due to a subsequent upstream reversal.
Of course, Alice herself might try to reverse the transaction that sent
Alice's tokens to the exchange, but as long as the exchange can show that 
it honestly processed the swap, that reversal request will be rejected,
and Alice will lose the fee needed to file a freeze request. 
Hence, the only change to the DeFi provider is that it may need to
respond to false reversal requests by Alice. 

When the exchange sends ERC-20R tokens to Alice,
it will call the {\em transfer()} function, as it does today.
This transfers funds from the exchange's \NRbalance.
Hence, the exchange will need to ensure 
that its \NRbalance is sufficiently high to support customer demand.  
This is done by calling the {\em clean()} function frequently,
to transfer maturing funds from the exchange's \Rbalance to it \NRbalance. 

Finally, some exchanges and lending protocols may be willing 
to take the risk and accept funds from Alice's \Rbalance,
most likely charging Alice a higher fee in the process.
These transactions will require greater scrutiny by the DeFi provider,
and will require additional software to properly assess the risk.


\paragraph{Reversibility is viral.}
An exchange that processes a swap of one 20R token for another 20R token is safe to do so without delay:
if one side of the transaction is later claimed to have come from stolen funds, 
the exchange can initiate a reversal on the other side of the transaction.
However, exchanging a 20R token for a non-reversible ERC-20 token is more dangerous.
The exchange might only approve such a swap after the dispute window has elapsed on the 20R side
of the transaction. 
This introduces a delay when swapping a reversible token for a non-reversible one.
Thus, once some key tokens become reversible, other tokens are incentivized to do the same
to eliminate this delay. 
In other words, reversibility is viral.

\paragraph{Automation.}
In addition to judges, one could also experiment with an algorithmic transaction monitor that factors into the final decision. A decision to freeze or reverse could then be a function of both the risk score given by the monitoring algorithm and the judges' votes. For example, if the algorithm determines a high likelihood that a transaction is a result of theft, then a reversal rejection would require a higher threshold of judges voting against the case. One could also use such an algorithm to broadcast alerts of  suspicious transactions, lowering the risk that a victim will miss the reversion window.

\paragraph{Appeals.}
Future work may consider an appeals process, in which the party who loses the trial can request another group of judges
to rule on the disputed transaction.   We do not explore this here.

\paragraph{Partial reversals.}
In some cases judges may opt for a partial reversal, where part of the funds go back to the owner,
and part stay with the recipient.  While we focused on complete reversals, the API can support
partial reversals.


\iftoggle{fullversion}{  
\paragraph{Grouping disputes.}
A transaction that sends coins from address $A_1$ to $A_2$ is often accompanied by a transaction from $A_2$ to $A_1$. 
For example, $A_1$ might send DAI to $A_2$, and in exchange $A_2$ sends an NFT to $A_1$. 
Suppose that both DAI and the NFT support our reversible standards.
Then, if $A_1$ asks to reverse the DAI transaction, $A_2$ will likely ask to reverse the NFT transaction.
As an optimization, both claims can be decided jointly by the same quorum of judges. 
}{}

\section{Conclusion} \label{sec:conclusion}

In this paper, we proposed extensions to ERC-20 and ERC-721 
that introduce a short time window when a transaction can be reversed. 
If adopted widely, these standards can protect the blockchain community from large financial losses.  
We hope that this paper will generate more discussion of reversibility 
as well as new designs to address the challenges that it raises.

\iftoggle{fullversion}{
\paragraph{\bf Acknowledgments.}
This work was supported by NSF, the Simons Foundation, NTT Research, Coinbase, and UBRI.
}{}

\appendix

\iftoggle{fullversion}{  
\section{Details of the freeze algorithm}
\label{sec:analysis}

In this section we extend and further analyze the freeze algorithm in Figure~\ref{fig:code}. 

\subsection{Properties of the algorithm}
\label{sec:opt}

First, let us examine some properties of the freeze algorithm in Figure~\ref{fig:code}. 
Consider the two transaction graphs in Figure~\ref{fig:two}.
Both show the flow of funds, starting with a disputed transaction from $v$ to $a_0$. 
Let us assume that when {\it freeze} is called, the balance at $a_0$ and $a_1$ is zero in both graphs.
As a result, in both graphs, calling {\it freeze} on the transaction from $v$ to $a_0$ 
will transfer the obligation to addresses $a_2$ and $a_3$, which is where the funds reside
at the time of the freeze. 

In the graph $G_1$ the transaction from $a_1$ to $a_3$ is the most recent,
and therefore the algorithm will freeze $10$ tokens at $a_3$ and zero tokens at $a_2$.
This may seem unfair because both $a_2$ and $a_3$ may have received a portion of the disputed funds.
We stress, however, that in this case there is no ``correct'' answer because it is not possible to determine 
how exactly $a_1$ split the 20 tokens at its disposal (10 from the disputed transaction and 10 from a prior balance). 
Our algorithm's decision is as ``correct'' as any other answer.  

Interestingly, an algorithm that tries to split the obligation evenly between $a_2$ and $a_3$ in the graph $G_1$ 
can lead to paradoxical results. 
Consider the graph~$G_2$.  
The algorithm needs to obligate $a_2$ and~$a_3$ to a total of 20 tokens.  
By the discussion in Section~\ref{sec:methodology} (Example~3), 
address $a_3$ bears the full responsibility for the {\em second} transaction from $a_0$ to $a_1$,
and will therefore be obligated to 10 tokens as a result of this transaction.
Now, if $a_3$ were to bear a portion of the responsibility for the {\em first} transaction from $a_0$ to $a_1$,
then the total obligation passed to $a_3$ would exceed 10 tokens. 
Since the balance at $a_3$ may only be equal to 10, 
this strategy will cause the algorithm to fail to freeze 20 tokens in total.
Consequently, splitting the obligation between $a_2$ and $a_3$ in such cases may lead
to freezing a reduced number of tokens, thereby preventing the victim from recovering the disputed funds. 

Our algorithm does not suffer from this issue:
Theorem~\ref{lem:amounts} below shows that our freeze algorithm in Figure~\ref{fig:code}
always freezes the correct amount: 10 tokens in the graph $G_1$ and 20 tokens in the graph~$G_2$.

\begin{figure}[t!]
\begin{center}
  \includegraphics[width=\linewidth]{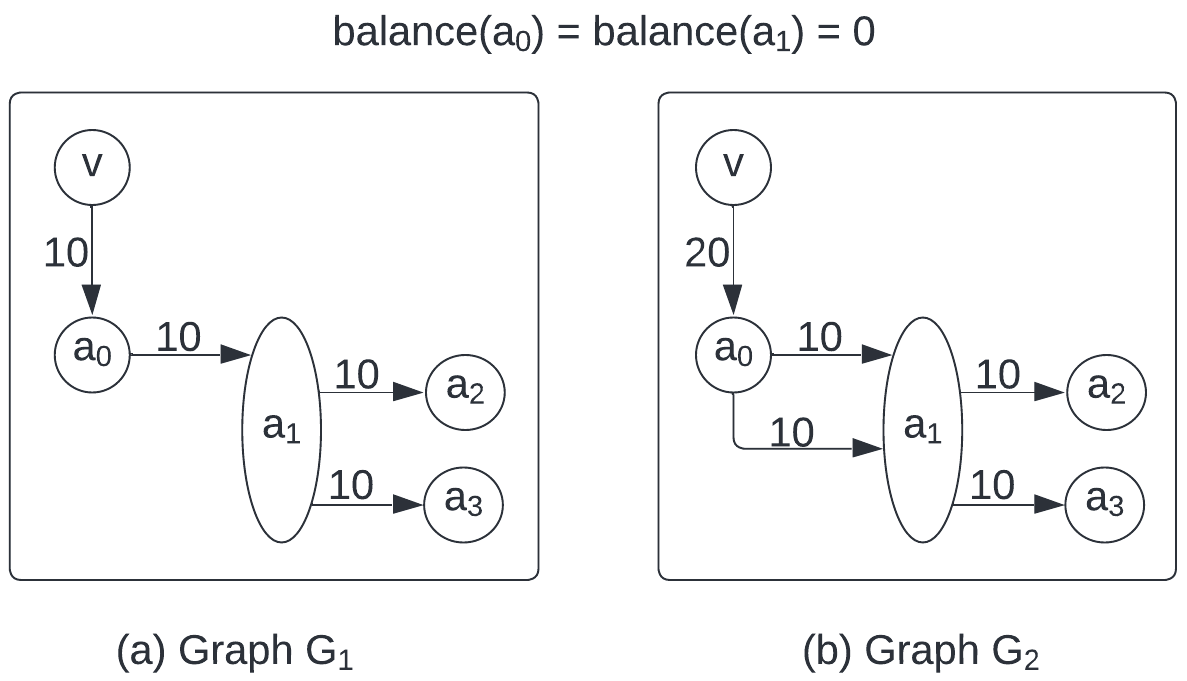}
\end{center} \vspace{-3ex}
  \caption{Two graphs where transactions are ordered chronologically from top to bottom.
     In both graphs, the balance at $a_0$ and $a_1$ is zero at the time of the freeze.
     In $G_1$ the starting balance at $a_1$ is 10.}
  \label{fig:two}
\rule{\linewidth}{1pt}
\end{figure}

\subsection{Graphs with cycles}
\label{sec:cycles}

The freeze algorithm in Figure~\ref{fig:code} relies on a topological sort
of the nodes in the graph~$G$.
The topological sort only exists when the graph is acyclic.  
Here we extend the algorithm to any directed graph, even one that includes cycles.
We use a pre-processing phase that eliminates cycles.  

Let us first see an example.
Suppose the graph contains a transaction $t = (a_0 \to a_1)$
for five coins, and a subsequent transaction $t' = (a_1 \to a_0)$
for three coins. 
This simple cycle means that $a_0$ sent five coins to $a_1$, and 
$a_1$ subsequently sent three coins back to $a_0$.  
For the purpose of our freeze algorithm, 
we can replace both transactions
by a single transaction from $a_0$ to $a_1$ for two coins, and eliminate
the cycle. 

More generally, let $t_0 = (v \to a_0)$ be the disputed transaction. 
The pre-processing step scans the graph 
rooted at $a_0$ to look for a directed cycle
of transactions $c_0, \ldots, c_{k-1}$ where $c_i = (b_i \to b_{i+1 \bmod k})$
that were all posted after the disputed transaction.
Let $c$ be the transaction of smallest value along the cycle,
breaking ties arbitrarily. 
Let $d \deq \val(c)$.
We can eliminate the cycle by removing transaction $c$ from the graph, 
and subtracting $d$ from the value of every other transaction in the cycle. 
The algorithm can repeat this process until the graph is free
of cycles.
At this point, the algorithm from Figure~\ref{fig:code} 
can be applied to the resulting graph.

\subsection{Two correctness theorems}

Finally, we argue that the algorithm will freeze the correct amount. 
The first theorem shows that the algorithm will freeze the correct total amount.
The second theorem shows that the algorithm will not over-freeze funds at any single address. 

\begin{theorem}  \label{lem:amounts}
Suppose that no burn transactions are processed between the disputed transaction $t_0$ and the freeze transaction $\FT$\,.
Moreover, there are no prior freezes in the system.
Then, if the disputed transaction $t_0$ sends $s$ tokens to address $a_0$,
then the freeze algorithm in Figure~\ref{fig:code} will freeze a total of exactly $s$ tokens.
In particular, when the algorithm terminates we have 
\begin{equation} \label{eq:lem2}
    s = \sum_{a} \frozen(a).
\end{equation}
\end{theorem}

\begin{proof}[Proof sketch]
First, let us assume that the directed graph rooted at $a_0$ is a DAG.
We argue that the algorithm satisfies two properties:
\begin{itemize}[itemsep=.5ex,topsep=.5ex,labelwidth=1em,leftmargin=1em]
\item {\em Property 1.}
When an address $a$ is processed in the main loop, 
let us call the quantity $\tau'$ computed on Line~\eqref{code:amount}
the {\bf excess obligation at $a$}, namely the obligation not covered by $a$'s current balance.
We claim that when the algorithm completes processing address~$a$, 
the quantity $\tau'$ is added to the total obligation of the children of~$a$. 

\item {\em Property 2.}
For a transaction $t: a \to b$, let $\textit{ob}(b,t)$ be the obligation that the algorithm passes to $b$ as a result of $t$.
We claim that for every transaction $t: a \to b$ we have
\[   \textit{ob}(b,t) \leq \val(t).  \]
This means that the funds that transaction $t$ sends to $b$ (and its descendants) is sufficient to cover the full obligation 
that the algorithm passes to $b$. 
\end{itemize}
Together these two properties prove the lemma: at every address, the exact excess obligation is passed to its children,
and there are sufficient funds at the descendants to cover these obligations. 

Let us prove these two properties.
Property~2 is immediate from  Line~\eqref{code:sub} in the algorithm.
This line ensures that for every transaction $t$, 
at most $\val(t)$ is added to the obligation at $b$. 

\begin{figure}[t!]
  \ \ \includegraphics[width=0.8\linewidth]{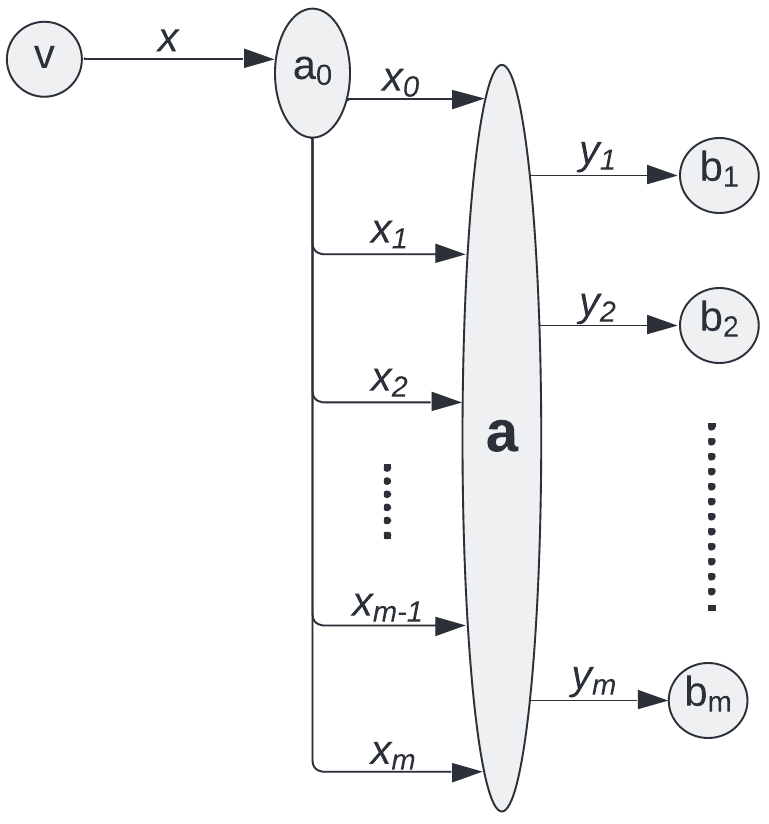} 
\vspace{-1ex}
  \caption{A sequence of $m+1$ incoming transactions to~$a$ from its parent $a_0$, 
           along with~$m$ outgoing transactions from $a$, interleaved in time.
           Incoming transaction number $i$ obligates $a_0$ to $x_i$ tokens, 
           and outgoing transaction number $j$ sends $y_j$ tokens to $b_j$.}
  \label{fig:quad}
\rule{\linewidth}{1pt}
\end{figure}

Property~1 follows from the fact that no coins are burned. 
In particular, consider the node $a$ in Figure~\ref{fig:quad}. 
The total obligation passed to $a$ is $\tau \deq x_0 + \ldots + x_m$. 
Let $\textit{Bal}(a)$ be the balance at $a$ just prior to the freeze transaction.
Then, by Property~2 and because no tokens were burned we know that
\[  
   \tau \le \textit{Bal}(a) + \underbrace{(y_1 + \ldots + y_m)}_{\text{\scriptsize amount that left $a$}}. 
\]
Moreover, by Line~\eqref{code:amount} we know that $\tau - \tau' = \textit{Bal}(a)$ and therefore
\begin{equation}  \label{eq:excess}
  \tau' \le y_1 + \ldots + y_m.  
\end{equation}
But when~\eqref{eq:excess} holds, the loop in Lines~\eqref{code:loop} to~\eqref{code:out} is guaranteed to add $\tau'$ 
to the total obligations of the children of $a$. 
This proves Property~1.

\medskip
To complete the proof we observe that the cycle elimination process from Section~\ref{sec:cycles}
does not affect the number of frozen tokens.
Hence, the lemma applies equally well to non-DAG graphs, assuming the cycle elimination process is applied first. 
\end{proof}

Next, we show that the algorithm will not over-obligate an address. 
This is where we rely on the reverse chronological order loop in Line~\eqref{code:loop}.
Consider again the transactions out of $a$ in Figure~\ref{fig:quad}.
We show that for every $j \in \{1,\ldots,m\}$, 
the address $b_j$ will bear no obligation for funds that arrived at $a$ after transaction number $j$.  
To define this, for a transaction $t: a \to b$ we use the following terminology:
\begin{itemize}[itemsep=.5ex,topsep=.5ex,labelwidth=1em,leftmargin=1em]
\item $\textit{ob}(b,t)$ is the obligation passed on to $b$ as a result of transaction $t$.
In other words, $\textit{ob}(b,t)$ is the value added to $\oblig(b)$ when the algorithm processes the transaction $t$.

\item $\textit{obsum}(a,t)$ is the sum of all the obligations that were added to $\oblig(a)$ as a result of the transactions
that sent funds to $a$ prior to $t$. 
For example, for a transaction $t_j: a \to b_j$ in Figure~\ref{fig:quad} we have 
\[  \textit{obsum}(a,t_j) = x_0+\ldots+x_{j-1}  \] 
\end{itemize}
\noindent
The following theorem shows that for a transaction $t:a \to b$ we have 
$\textit{ob}(b,t) \le \textit{obsum}(a,t)$.
This means that the obligation passed to $b$ will not exceed the obligations passed to $a$ due to transactions that preceded $t$.
In other words, $b$ will not bear responsibility for an amount transferred to $a$ after transaction~$t$.

\begin{theorem}  \label{thm:schedule}
Suppose that no burn transactions are processed between the disputed transaction $t_0$ and the freeze transaction $\FT$\,.
Then when the algorithm terminates, for every transaction $t: a \to b$ we have
\begin{equation} \label{eq:lem3}
    \textit{ob}(b,t) \le \textit{obsum}(a,t)
\end{equation}
\end{theorem}
\begin{proof}
Let us prove the theorem for the graph in Figure~\ref{fig:quad}.
The exact same argument applies to other graphs. 
Let $\textit{Bal}(a)$ be the balance at address $a$ at the time of the freeze.
As usual, let
\begin{equation} \label{eq:tau}
  \tau = x_0 + \ldots + x_m.   
\end{equation}
Now, fix some $j \in \{1,\ldots,m\}$ and let $t_j$ be the transaction $a \to b_j$ in Figure~\ref{fig:quad}.
Observe that if
\[  \tau \le \textit{Bal}(a) + \sum_{i=j+1}^m y_i \]
then $\textit{ob}(b_j,t_j) = 0$ and the theorem follows trivially. 
Hence, we can assume that
\[  \tau > \textit{Bal}(a) + \sum_{i=j+1}^m y_i. \]
In this case, 
by definition of the loop in Lines~\eqref{code:loop} to~\eqref{code:out} in the algorithm we know that
\begin{multline} \label{eq:full}
   \textit{ob}(b_j,t_j) = \\
       \min\left( y_j,\ \ \  \tau - \textit{Bal}(a) - \sum_{i=j+1}^m y_i \right). 
\end{multline}
Since no coins were burned at $a$ we know that
\[  \sum_{i=j}^m x_i \le \textit{Bal}(a) + \sum_{i=j+1}^m y_i   \]
and therefore by~\eqref{eq:full} we obtain
\[   \textit{ob}(b_j,t_j) \le \min\bigl( y_j,\  \tau - \sum_{i=j}^m x_j  \bigr).  \]
It now follow by~\eqref{eq:tau} that
\[  \textit{ob}(b_j,t_j)  \le x_0 + \ldots + x_{j-1} = \textit{obsum}(a,t_j)   \]
as required.  This completes the proof.
\end{proof}

\subsection{Preventing double freezes}
\label{sec:double}

Consider two transactions $v \to a_0$ and $a_0 \to a_1$
issued in that order, each transferring~$\tau$ tokens.
Suppose, that the final balance at $a_0$ is zero and at $a_1$ is $2 \tau$. 
Next, suppose $v$ requests to freeze the $v \to a_0$ transaction,
which will freeze~$\tau$ tokens at $a_1$.
The problem is that now $a_0$ can request to freeze the $a_0 \to a_1$ transaction.
This second request is attempting to again freeze the $\tau$ tokens from~$v$,
and should do nothing.
However, our freeze algorithm, as currently described, 
will freeze an additional $\tau$ tokens at $a_1$. 

We briefly describe two ways to prevent double freezes. 
First, the judges can recognize the second freeze
as a double freeze and reject this freeze request. 

Alternatively, there is a simple modification to the 
freeze algorithm in Figure~\ref{fig:code}
that eliminates double freezes. 
Recall that the freeze algorithm in Figure~\ref{fig:code}
processes one transaction at a time in the transaction graph.
Every transaction is recorded in the 
$\spenditures$ data structure shown in Figure~\ref{fig:spenditures}.
Consider a transaction $a \to b$, and let 
$S = (\textit{to}, \textit{amount}, \textit{blockNumber})$ 
be the corresponding triple in the $\spenditures$ data structure.
If the freeze algorithm obligates $b$ to $\tau_0 \le \tau$ tokens
due to this transaction, then the algorithm does two more steps:
first, it subtracts $\tau_0$ from the {\it amount} field in the triple $S$,
and second, it adds the triple $(a, \textit{index}, \tau_0)$
to the {\tt \_claims} array for this freeze request. 
The first step ensures that funds cannot be double frozen.
The second step is needed for the {\it rejectReverse()} function.
If the reversal is eventually rejected, then the {\it rejectReverse()}
function will add $\tau_0$ 
back to the triple $S$ in the $\spenditures$ data structure
from which it was subtracted.

}{}

\bibliography{whitepaper}
\bibliographystyle{plain}

\end{document}